\documentclass[11pt]{article}

\usepackage[T1]{fontenc}
\usepackage[utf8]{inputenc}
\usepackage{microtype}
\usepackage{fullpage}
\usepackage{comment}

\usepackage{amsmath,amssymb,amsthm}

\usepackage{graphicx}
\usepackage{subcaption}
\usepackage{booktabs}
\usepackage{tikz}
\usetikzlibrary{calc,positioning,arrows.meta,shapes}

\usepackage{algorithm}
\usepackage[noend]{algpseudocode}

\usepackage{enumitem}

\usepackage[normalem]{ulem} 

\usepackage{xcolor}

\newtheorem{theorem}{Theorem}
\newtheorem{lemma}[theorem]{Lemma}
\newtheorem{proposition}[theorem]{Proposition}
\newtheorem{corollary}[theorem]{Corollary}
\newtheorem{claim}[theorem]{Claim}
\newtheorem{observation}[theorem]{Observation}

\theoremstyle{definition}

\title{The Minimum Subgraph Complementation Problem}

\author{
Juan Gutiérrez\thanks{Departamento de Ciencia de la Computación,
Universidad de Ingeniería y Tecnología (UTEC), Lima, Perú.
\texttt{jgutierreza@utec.edu.pe}. This work was supported by Fondo Semilla UTEC 2025.}
\and
Sagartanu Pal\thanks{School of Computer Science Engineering and Technology,
Bennett University, India.
\texttt{sagartanu.pal@bennett.edu.in, maxfi1987@gmail.com}}
}

\date{}

\begin{document}
\maketitle

\begin{abstract}
Subgraph complementation is an operation that toggles all adjacencies within a selected vertex set. Given a graph $G$, for target class $\mathcal{C}$, the Minimum Subgraph Complementation (MSC) problem asks for a smallest vertex set whose complementation transforms $G$ into a graph in $\mathcal{C}$. While the decision version is NP-complete for many graph classes, the algorithmic complexity of the optimization variant remains largely unexplored.

We initiate a systematic algorithmic study of MSC. We present polynomial-time algorithms for several nontrivial settings, including transformations between bipartite, co-bipartite, and split graphs, complementing 2-connected bipartite regular graphs into chordal graphs, and forests into graphs of bounded degeneracy. We also study connectivity-based targets, showing that MSC can be solved in polynomial time for arbitrary graphs when the goal is to obtain disconnected or 2-connected graphs. Our results rely on structural characterizations and classical algorithmic tools, providing a foundation for the algorithmic study of MSC.
\end{abstract}

\section{Introduction}
In this work, a graph is always simple and undirected, and the notation used is standard~\cite{BondyM08,Diestel2017}.
Given a graph $G$ and a nonempty subset $S$ of vertices of $G$, the graph $G \oplus S$ is obtained from $G$ by complementing the adjacencies between vertices of $S$, that is, by replacing the induced subgraph $G[S]$ with its complement.
Thus, when $S$ consists of two vertices, the operation simply toggles the adjacency between them, while for larger sets $S$, it reverses all adjacencies among the vertices of $S$.
This operation, known as subgraph complementation, was first introduced by Kamiński et al.~\cite{DBLP:journals/dam/KaminskiLM09} in the context of clique-width.
A more systematic algorithmic investigation was later undertaken by Fomin et al.~\cite{DBLP:journals/algorithmica/FominGST20}, who studied the following decision problem.

\medskip
\fbox{%
\begin{minipage}{0.95\linewidth}
\textbf{Subgraph Complementation to $\mathcal{C}$ (SC$\mathcal{C}$)}\\
\textbf{Input:} A graph $G$.\\
\textbf{Question:} Does there exist a subset $S \subseteq V(G)$ such that $G \oplus S \in \mathcal{C}$?
\end{minipage}%
}

Fomin et al.~\cite{DBLP:journals/algorithmica/FominGST20} initiated the systematic study of the Subgraph Complementation problem and established several tractability and hardness results. In particular, they proved that SC is polynomial-time solvable for every triangle-free graph class $\mathcal{C}$ recognizable in polynomial time. They also showed that SC is polynomial-time solvable for every $(k,\ell)$-graph class with $k+\ell=2$, namely bipartite, split, and co-bipartite graphs. In addition, they obtained polynomial-time algorithms for $d$-degenerate and $P_4$-free graphs, while proving NP-completeness when $\mathcal{C}$ is the class of regular graphs. Moreover, they explicitly posed the complexity of SC to chordal graphs as an open problem.

Subsequent studies~\cite{DBLP:journals/algorithmica/AntonyGPSSS22,antony2024cutting,DBLP:journals/ipl/AntonyPS25} further refined these boundaries. In particular, they showed that SC becomes NP-complete when $\mathcal{C}$ is the class of graphs excluding cycles of length at least seven, the class of graphs excluding certain trees, or the class of $H$-free graphs for a 5-connected non-self-complementary prime graph $H$ with at least 18 vertices. A broader survey of these and related edge-modification problems can be found in~\cite{DBLP:journals/csr/CrespelleDFG23}.

\begin{table*}[ht]
\centering
\small

\begin{minipage}[t]{0.42\textwidth}
\centering
\textbf{SC Results}

\vspace{2mm}

\begin{tabular}{lll}
\toprule
Target class & Complexity & Reference \\
\midrule
Triangle-free recognizable & P & \cite{DBLP:journals/algorithmica/FominGST20} \\
$d$-degenerate graphs & P & \cite{DBLP:journals/algorithmica/FominGST20} \\
$P_4$-free graphs & P & 
\cite{DBLP:journals/algorithmica/FominGST20} \\
$(k,\ell)$-graphs with $k + \ell =2$ & P & 
\cite{DBLP:journals/algorithmica/FominGST20} \\
Regular graphs & NP-complete & \cite{DBLP:journals/algorithmica/FominGST20} \\
Cycle-free classes & NP-complete & \cite{DBLP:journals/ipl/AntonyPS25} \\
Tree-free classes & NP-complete & \cite{antony2024cutting} \\
Certain $H$-free classes & NP-complete & \cite{DBLP:journals/algorithmica/AntonyGPSSS22} \\
Chordal graphs & Open & \cite{DBLP:journals/algorithmica/FominGST20} \\
\bottomrule
\end{tabular}
\end{minipage}
\hfill
\hspace{1cm}
\hfill
\begin{minipage}[t]{0.42\textwidth}
\centering
\textbf{MSC Results}

\vspace{2mm}

\begin{tabular}{lll}
\toprule
Input class & Target class & Complexity \\
\midrule
Bipartite & Co-bipartite & P \\
Split & Bipartite & P \\
Bipartite & Split & P \\
Co-bipartite & Split & P \\
Biregular bipartite & Chordal & P \\
Forest & $D_k$ & P \\
Arbitrary & 2-connected & P \\
Arbitrary & Disconnected & P \\
\bottomrule
\end{tabular}
\end{minipage}

\caption{Known complexity results for the SC problem and polynomial-time solvable MSC cases established in this paper.}
\label{tab:summary-results}
\end{table*}

In this paper, we adopt a different perspective.  
For many natural graph classes, the decision variant above is trivial, as every input graph admits a subset $S$ such that $G \oplus S$ belongs to the target class.  
In such cases, the relevant algorithmic question is not whether a valid subset exists, but rather how small such a subset can be.  
This motivates the following optimization version, which we call the Minimum Subgraph Complementation problem.

\medskip
\fbox{%
\begin{minipage}{0.95\linewidth}
\textbf{Minimum Subgraph Complementation to $\mathcal{C}$(MSC$\mathcal{C}$)}\\[1mm]
\textbf{Input:} A graph $G$.\\[0.5mm]
\textbf{Output:} A subset $S \subseteq V(G)$ of minimum size such that $G \oplus S \in \mathcal{C}$, 
or report that no such subset exists.
\end{minipage}%
}
\medskip

We note that MSC, viewed as an optimization problem, is NP-hard in general.
Indeed, MSC is at least as hard as the corresponding SC
decision problem, since an algorithm that computes an optimal solution would
also decide whether there exists a solution.
Therefore, for graph classes for which the SC is NP-complete
\cite{DBLP:journals/algorithmica/FominGST20,
DBLP:journals/algorithmica/AntonyGPSSS22,
antony2024cutting},
MSC is NP-hard.
For convenience, Table~\ref{tab:summary-results} summarizes the main known complexity results for the SC problem together with the new MSC results established in this work.

Studying this optimization version is meaningful even for pairs of classes $(\mathcal{G}, \mathcal{C})$ where every graph $G \in \mathcal{G}$ can be complemented into a member of $\mathcal{C}$.  
Indeed, the complexity and structure of a minimum complementing set $S$ can vary significantly, revealing interesting combinatorial properties of the classes involved.  
For instance, while every bipartite graph can be complemented into a chordal graph by complementing one side of its bipartition, the size and structure of the minimum set $S$ can differ substantially from graph to graph.

Graph editing problems have been extensively studied as a way of measuring how far a given graph is from satisfying a desired property. 
In these problems, one seeks to modify a given graph by a minimum number of operations, such as edge additions or deletions, so that the resulting graph belongs to a target class. 
Such modifications provide a natural measure of the ``distance'' of a graph to a property of interest and have been investigated from both classical and parameterized complexity perspectives~\cite{AlonShapiraSudakov2005,CrespelleDrangeFominGolovach2023,DzidoKrzywdzinski2015,Liu2015EdgeDeletion,NatanzonShamirSharan2001}. 

For instance, the Chordal Editing problem, which asks whether a graph can be made chordal by a minimum number of edge modifications, has received considerable attention due to the central role of chordal graphs in structural graph theory and algorithmic applications~\cite{CaoMarx2014}. 
These studies illustrate that even when every graph in a certain class can be trivially transformed into the target class, the size and structure of a minimum modifying set can vary significantly, revealing interesting combinatorial properties of the graph classes involved.
In this context, the MSC problem can be regarded as a structured variant of edge editing, where instead of modifying arbitrary edges, one flips all edges and non edges inside a single vertex subset.

In this paper, we show positive results for the MSC problem in particular scenarios. 
To introduce the topic to the reader, we begin with an easy theorem about bipartite graphs. 
To motivate this particular case, observe that for any class of graphs $\mathcal{G}$, 
every graph $G$ whose complement belongs to $\mathcal{G}$ can be complemented into $\mathcal{G}$ 
by taking $S = V(G)$. However, the set $V(G)$ does not necessarily yield a minimum-size solution 
for complementability into $\mathcal{G}$.

For example, let $\mathcal{G}$ be the class of planar graphs and consider the nonplanar graph $K_5$. 
If we take a set $S = \{u,v\}$ consisting of the endpoints of any edge $uv$ of $K_5$, 
then toggling the edges with both endpoints in $S$ removes exactly the edge $uv$. 
The resulting graph $K_5 - uv$ is planar. Thus, $K_5$ is complementable to $\mathcal{G}$ using a set $S$ 
of size~2, whereas the trivial bound $|S| = |V(K_5)| = 5$ is far from optimal.


This motivates the following particular case: solving the MSC problem when the input graph is not in $\mathcal{G}$. 
In what follows, we show that this problem can be solved efficiently when 
$\mathcal{G}$ is the class of bipartite graphs or the class of co-bipartite graphs.

\begin{theorem}
\label{thm:bipartite-cobipartite}
Let $\mathcal{C}$ be the class of co-bipartite graphs. 
Then MSC to $\mathcal{C}$ can be solved in polynomial time when the input graph is bipartite.
\end{theorem}
\begin{proof}
Let $G$ be a bipartite graph with bipartition $(A,B)$.
Let $S \subseteq V(G)$ such that $G \oplus S$ is co-bipartite.  
Assume $S \cap A \neq \emptyset$; we show that $|A \setminus S| \le 1$.  

Suppose, for contradiction, that $|A \setminus S| \ge 2$, and pick any $w \in S \cap A$.  
Since $G$ is bipartite, the vertices of $(A \setminus S) \cup \{w\}$ form an independent set in $G$.  
Toggling edges inside $S$ does not introduce edges in this set, so it remains independent in $G \oplus S$.  
Hence $G \oplus S$ contains an independent set of size at least three, contradicting that it is co-bipartite.  
Thus, either $S \cap A = \emptyset$ or $|A \setminus S| \le 1$.  

The argument for $B$ is symmetric.  
Therefore, every feasible $S$ satisfies $S \cap B = \emptyset$ or $|B \setminus S| \le 1$.

It follows that it suffices to enumerate all subsets $A' \subseteq A$ and $B' \subseteq B$ such that
$|A'| \ge |A|-1$ and $|B'| \ge |B|-1$, forming candidate sets $S \in \{A' \cup B', A', B'\}$.  
There are polynomially many such candidates, and for each we can test in polynomial time whether $G \oplus S$ is co-bipartite.  
A smallest feasible $S$ is thus optimal.
\end{proof}

Now, although the case where $G$ is co-bipartite and $\mathcal{C}$ is the class of bipartite graphs is fully analogous, we introduce a lemma that allows us to derive this result directly from Theorem~\ref{thm:bipartite-cobipartite}.

\begin{lemma}
\label{lem:complement-transfer}
Let $\mathcal{C}$ and $\mathcal{D}$ be two graph classes.
If MSC to $\mathcal{C}$ can be solved in polynomial time
when the input graph belongs to $\mathcal{D}$, then
MSC to $\overline{\mathcal{C}}$ can also be solved in polynomial time
when the input graph belongs to $\overline{\mathcal{D}}$.
\end{lemma}

\begin{proof}
We first observe the following claim, already proven in Proposition 3 of \cite{DBLP:journals/algorithmica/AntonyGPSSS22}.
\begin{claim}
\label{claim:complement-commutes}
For every graph $G$ and every set $S \subseteq V(G)$,
$G \oplus S = \overline{\overline{G} \oplus S}$.
\end{claim}

\begin{proof}[Proof of Claim~\ref{claim:complement-commutes}]
Let $u,v \in V(G)$.
If $u,v \in S$, then their adjacency is toggled in both
$G \oplus S$ and $\overline{G} \oplus S$.
If at least one of $u$ or $v$ is not in $S$, then their adjacency is
unchanged in both graphs.
Hence $uv$ is an edge of $G \oplus S$ if and only if it is a non-edge of
$\overline{G} \oplus S$, proving the claim.
\end{proof}
Now let $G \in \overline{\mathcal{D}}$.
Then $\overline{G} \in \mathcal{D}$.
By assumption, we can compute in polynomial time a minimum-size set
$S \subseteq V(G)$ such that $\overline{G} \oplus S \in \mathcal{C}$.
By Claim~\ref{claim:complement-commutes},
$\overline{G} \oplus S = \overline{G \oplus S}$, and therefore
$G \oplus S \in \overline{\mathcal{C}}$.
Thus $S$ is a feasible solution for $G$ with respect to
$\overline{\mathcal{C}}$.
The converse direction follows symmetrically.
\end{proof}

The following corollary follows directly from Theorem~\ref{thm:bipartite-cobipartite} and Lemma~\ref{lem:complement-transfer}.

\begin{corollary}
\label{cor:co-bipartite-bipartite}
Let $\mathcal{C}$ be the class of bipartite graphs. 
Then MSC to $\mathcal{C}$ can be solved in polynomial time when the input graph is co-bipartite.
\end{corollary}

Bipartite and co-bipartite graphs are particular cases of a more general class called $(k, \ell)$-graphs, which we consider in Section~\ref{section:kl}.  
Our main results in this section show that the MSC problem can be solved efficiently for these graph classes:

\begin{itemize}
    \item When the target class is the class of bipartite graphs and the input is a split graph (Theorem~\ref{thm:split-bipartite}).
    \item When the target class is the class of co-bipartite graphs and the input is a split graph (Corollary~\ref{cor:split-co-bipartite}).
    \item When the target class is the class of split graphs and the input is a bipartite graph (Theorem~\ref{thm:bipartite-split}).
    \item When the target class is the class of split graphs and the input is a co-bipartite graph (Corollary~\ref{cor:co-bipartite-split}).
\end{itemize}

Note that every split graph is also a chordal graph.  
Chordal graphs form a superclass of many well-behaved graph families, such as interval and split graphs, and allow efficient algorithms for problems that are NP-hard in general, including coloring, maximum clique, and recognition of minimal separators.  
Hence, it is natural to ask whether the MSC problem can be solved efficiently when the target class is chordal graphs and the input is bipartite.  
In Section~3, we show that this is indeed the case when the input graph is 2-connected biregular (Theorem~\ref{thm:biregular-chordal}).

From a structural perspective, every chordal graph is also $k$-degenerate, where $k$ equals its maximal clique size minus one.  
The study of editing or modification problems in $k$-degenerate graphs has been considered in the literature~\cite{Mathieson2010}, so it is natural to try to generalize results on bipartite graphs to graphs of degeneracy $k$.  
In Section~\ref{section:forest-deg-k}, we show that this case can be solved in polynomial time when the input is a forest (Theorem~\ref{thm:forest-degenerate}).

The last two sections of this work focus on modifying the connectivity of a graph.  
The problem of modifying a graph through local operations in order to increase its connectivity has been studied under several models.  
For instance, in the Connectivity Augmentation problem, the goal is to add the smallest possible set of edges so that a given graph becomes $k$-connected; polynomial-time algorithms are known for small fixed values of $k$, while the general problem is NP-hard~\cite{EswaranTarjan1976,WatanabeNakamura1987,ByrkaGrandoniAmeli2019}.  
Also, 
problems aiming at destroying the connectivity of a graph have been studied before under the framework of edge- and vertex-deletion problems.
In particular, Yannakakis~\cite{Yannakakis1981} investigated the complexity of edge-deletion problems where the objective is to eliminate a given graph property, including connectivity.
This line of work provides a natural conceptual background for our results on MSC when the target class is the class of disconnected graphs.

Our main results show that the MSC problem can be solved efficiently in these settings:
\begin{itemize}
    \item When the target class is the class of 2-connected graphs and the input graph is arbitrary (Theorem~\ref{thm:to-2conn}).
    \item When the target class is the class of disconnected graphs and the input graph is arbitrary (Theorem~\ref{thm:to-0conn}).
\end{itemize}

\section{$(k,\ell)$-graphs}\label{section:kl}

Bipartite and cobipartite graphs are particular cases of a more general definition.
A graph $G$ is a \((k,\ell)\)-graph if its vertex set admits a partition
\[
V(G) = C_1 \cup \dots \cup C_k \;\cup\; I_1 \cup \dots \cup I_\ell
\]
such that each $C_i$ induces a clique in $G$ and each $I_j$ induces an 
independent set in $G$.

The classes of bipartite and cobipartite graphs can be described in terms of 
$(k,\ell)$-graphs. In particular, $(0,2)$-graphs are exactly the bipartite graphs 
and $(2,0)$-graphs are exactly the co-bipartite graphs. A graph $G$ is a 
\emph{split graph} if and only if it is a $(1,1)$-graph: in this case, $V(G)$ can 
be partitioned into one clique and one independent set.

The next observation is immediate.

\begin{observation}
Every $(k,\ell)$-graph is complementable to any $(k',\ell')$-graph whenever 
$k+\ell = k'+\ell'$.
\end{observation}

Focusing on the case $k+\ell = 2$, this implies that any bipartite or 
cobipartite graph is complementable to a split graph, and that any split graph 
is complementable to a bipartite or a cobipartite graph. The cases 
$(k,\ell),(k',\ell') \in \{(0,2),(2,0)\}$ were solved in 
Theorem~\ref{thm:bipartite-cobipartite}. 
In the remainder of this section, we focus on solving the four remaining cases.

\begin{theorem}
\label{thm:split-bipartite}
Let $\mathcal{C}$ be the class of bipartite graphs.  
Then MSC to $\mathcal{C}$ can be solved in polynomial time when the input graph is a split graph.
\end{theorem}

\begin{proof}
Let $G$ be a split graph with partition $V(G)=K\cup I$, where $K$ is a clique and $I$ is an independent set.  
Consider a set $S\subseteq V(G)$ such that $G\oplus S$ is bipartite.

If $S\cap K\neq\emptyset$, then at most one vertex of $K$ may lie outside $S$.  
Indeed, suppose two vertices $u,v\in K$ lie outside $S$ and some vertex $w\in K$ lies in $S$.  
Since all edges of $K$ are present in $G$, every edge among $\{u,v,w\}$ has at least one endpoint outside $S$, and hence all three edges survive in $G\oplus S$.  
Thus $\{u,v,w\}$ induces a triangle in $G\oplus S$, contradicting bipartiteness.  
Hence $|K\setminus S|\le 1$.

Similarly, if $S\cap I\neq\emptyset$, then $S\cap I$ contains at most two vertices.  
Otherwise, if three vertices of $I$ lie in $S$, then all three edges among them appear after complementation, forming a triangle in $G\oplus S$, again contradicting bipartiteness.  
Thus $|S\cap I|\le 2$.

Therefore any feasible solution $S$ that intersects both $K$ and $I$ must satisfy
\[
|K\setminus S|\le 1 
\qquad\text{and}\qquad
|S\cap I|\le 2.
\]
Consequently, every candidate solution $S$ is of one of the following forms:
$
S = K', \qquad
S = I', \qquad
S = K' \cup I',
$
where $K'\subseteq K$ and $I'\subseteq I$, and where
$
|K'|\ge |K|-1 
\qquad\text{and}\qquad 
|I'|\le 2.
$
There are $O(|K|)$ possibilities for $K'$ and $O(|I|^2)$ possibilities for $I'$.  
For each candidate we test in polynomial time whether $G\oplus S$ is bipartite, and return one of minimum size.  
This yields a polynomial-time algorithm for MSC to $\mathcal{C}$ when $\mathcal{C}$ is the class of bipartite graphs.

\end{proof}

As the class of split graphs coincides with the class of co-split graphs,
the following corollary follows directly from Theorem~\ref{thm:split-bipartite}
and Lemma~\ref{lem:complement-transfer}.

\begin{corollary}
\label{cor:split-co-bipartite}
Let $\mathcal{C}$ be the class of co-bipartite graphs.
Then MSC to $\mathcal{C}$ can be solved in polynomial time
when the input graph is a split graph.
\end{corollary}

Finally, we consider the case where the target class is the class of split graphs and the input is a bipartite graph.  
This case requires more work than the previous ones, and we will need some auxiliary lemmas.
A first observation is that we can ignore isolated vertices in the input graph. Indeed, note that a solution of minimum size cannot contain isolated vertices, as we can remove such a vertex from the solution, contradicting its minimum size. 

From now on, we fix a bipartite graph $G$ that is not a split graph and does not have isolated vertices, and a bipartition
$(A,B)$ of $G$ with $|A|,|B|>1$.
We call a subset of vertices a solution if its complementation makes the graph a split graph.
A solution $S$ is \textit{special} if $|S| > 2$ and either $|S \cap B| = 1$ or $|S \cap A| = 1$.
Let $Z$ be the set of vertices of $B$ that has neighbors with degree one. For each $z \in Z$, let $O_z$ denote the set of such neighbors and let $O = \bigcup_{z\in Z}O_z$.

\begin{lemma}\label{lemma:solving-special-solutions}
We can decide whether
$G$ admits a special solution that is optimal, and if it exists, we can find one in polynomial time.
\end{lemma}
\begin{proof}

    Let $S$ be a special solution that is also optimal.
    Let $(K,I)$ be a split partition of $G'=G \oplus S$.
We may assume without loss of generality that $|S \cap B|=1$.
As $A$ is a solution and $S$ is optimal, the next claim is clear.
      \begin{claim}\label{claim:A-not-subset-S}
         $A \not \subseteq S$.
    \end{claim}

    \begin{claim}\label{claim:uv-in-A-cap-K-then-uv-in-S}
        If $u,v \in A \cap K$ or $u,v \in B \cap K$, then $u,v \in S$.
    \end{claim}
    \begin{proof}
        Suppose $u,v \in A \cap K$. 
        As $u,v \in A$, we have $uv \notin E(G)$.
        As $u,v \in K$, we have $uv \in E(G')$.
        Thus $u,v \in S$. The proof when $u,v \in B \cap K$ is similar.
    \end{proof}

    \begin{claim}\label{claim:varios}
        $|S\cap A\cap I| \leq 1$, $|S\cap A \cap K| \geq 1$, $A \setminus S \subseteq I$ and $|B \cap K| \leq 1$.
    \end{claim}
    \begin{proof}
        Suppose for a moment that there exists $u,v \in S\cap A \cap I$. In that case, as $A$ is an independent set, $uv \in E(G')$, a contradiction. Hence $|S\cap A\cap I| \leq 1$. And, as $|S\cap A|>1$, we must have $|S\cap A \cap K| \geq 1$.
        Now, as there are no edges between vertices in $A \cap S$ and $A \setminus S$, we must have  $A \setminus S \subseteq I$.
        Finally, as $|B \cap S| = 1$, if there are two vertices in $B \cap K$, at least one of them is not in $S$, so there is no edge between such vertex and any other vertex in $B \cap K$, a contradiction.
    \end{proof}
    
    \begin{claim}\label{claim:A-set-minus-O-subset-S}
         $A \setminus O \subseteq S$.
    \end{claim}
    \begin{proof}
         Suppose, for a contradiction, that there exists $u \in A \setminus (O \cup S)$.
         Thus, as $u \notin O$ and $G$ has no isolated vertices, $u$ has two neighbors $x,y$.
         By Claim \ref{claim:varios}, $u \in I$. So $x,y \in B \cap K$.
          By Claim \ref{claim:uv-in-A-cap-K-then-uv-in-S}, $x,y \in S$, a contradiction to the assumption that $S$ is special.
    \end{proof}
    
    \begin{claim}\label{claim:z-in-ZcapI-then-Oz-subseteq-S}
            If $z \in Z \cap I$ then $O_z \subseteq S$.
    \end{claim}
    \begin{proof}
    Suppose by contradiction that there exists $u \in O_z \setminus S$.
        By Claim \ref{claim:varios}, there exists a vertex $x \in A \cap S \cap K$. As $xu \notin E(G')$, because $u \notin S$, we must have $u \in I$.
        But then, as $uz \in E(G')$ and $z \in I$, we have a contradiction.
       \end{proof}

     \begin{claim}\label{claim:z-in-ZcapK-then-Oz-cap-S-emptyset}
       If $z \in Z \cap K$ then  $O_z \cap S = \emptyset$.
    \end{claim}
    \begin{proof}
    Let $x \in O_z$ and suppose by contradiction $x \in S$.
    Now consider the set $S'= S-x$
    and the sets $K'=K \setminus \{x\}$ and $I'=I \cup \{x\}$. (Observe that if $x \in I$ then $(I,K)=(I',K')$).
        Note that, after complementing $S'$, as $x\in I'$, no edges between vertices of $K'$ are lost.
         Moreover, as the only neighbor of $x$ is $z$, which is in $K'$, no edges between vertices of $I'$ were added. 
Thus, $(K',I')$ is a split partition of $G \oplus S'$, contradicting the minimality of $S$. 
    \end{proof} 

    \begin{claim}\label{claim:Z-setminus-S-leq-1}
       $|Z \setminus S| \leq 1$.
     \end{claim}
     \begin{proof}
         Suppose by contradiction that $|Z \setminus S| \geq 2$.
Let $v \in Z \setminus S$ and let $x \in O_v$. Suppose for a moment that $x \notin S$. By Claim \ref{claim:varios}, $x \in I$. As $v \notin S$, then $xv \in E(G')$. So $v \in K$. As $v$ was arbitrary, every other vertex of $Z \setminus S$ is also in $K$. This implies that $|B \cap K| \geq 2$, a contradiction to Claim \ref{claim:varios}. Thus, every such $x$ is in $S$,
which implies that $O_v \subseteq S$. As $v$ was arbitrary, we must have $O \subseteq S$ and thus $A \subseteq S$ by Claim \ref{claim:A-set-minus-O-subset-S}, contradicting Claim \ref{claim:A-not-subset-S}.
     \end{proof}
    
We now continue with the proof of the lemma.
If $|Z|=0$ then, by Claim \ref{claim:A-set-minus-O-subset-S}, $A \subseteq S$ and we contradict Claim \ref{claim:A-not-subset-S}.
If $|Z| \geq 3$, then, as $|S \cap B|=1$,
we must have $|Z \setminus S| \geq 2$, a contradiction to Claim \ref{claim:Z-setminus-S-leq-1}.
Thus $|Z| \in \{1,2\}$.
Now consider the following algorithm. First we find $Z$, answering no if $|Z| \notin \{1,2\}$.
If $|Z|=1$, then we set $z$ to the only vertex in $Z$
and ask if $\{v\} \cup (A \setminus O_z)$ is a valid solution for some $v \in B$. If that is the case we return such solution.
If $|Z|=2$, then we let $v$ and $w$ be the two vertices of $Z$ and ask if one of $\{v\} \cup (A \setminus O_v)$, $\{v\} \cup (A \setminus O_w)$ or $\{v\} \cup (A \setminus O)$ is a valid solution. If that is the case we return one of minimum size.
If no valid solution is found, the algorithm returns \textsc{No}.

Now we show correctness. Let $S$ be a special solution that is also optimal. As such a solution exists, we know $|Z| \in \{1,2\}$. First suppose $|Z|=1$ and let $z$ be the only vertex in $Z$. By Claims \ref{claim:z-in-ZcapI-then-Oz-subseteq-S} and \ref{claim:z-in-ZcapK-then-Oz-cap-S-emptyset}, either $O_z \subseteq S$ or $O_z \cap S = \emptyset$.
As $A \setminus O \subseteq S$ by Claim \ref{claim:A-set-minus-O-subset-S} and $A \not\subseteq S$ by Claim \ref{claim:A-not-subset-S}, we must have $O_z \cap S = \emptyset$. Hence $ S = \{v\} \cup (A \setminus O_z)$ for some $v \in B$. As we iterated over all $v \in B$, the solution given by the algorithm has the same size as $S$.
Now suppose $|Z|=2$ and let $Z=\{v,w\}$.
By Claims \ref{claim:z-in-ZcapI-then-Oz-subseteq-S} and \ref{claim:z-in-ZcapK-then-Oz-cap-S-emptyset}, either $O_v \subseteq S$ or $O_v \cap S = \emptyset$, and either $O_w \subseteq S$ or $O_w \cap S = \emptyset$.
Also, by Claim \ref{claim:Z-setminus-S-leq-1}, one of $\{v,w\}$ is in $S$.
If $v \in S$ then $S$ is one of $\{v\} \cup (A \setminus O_v)$, $\{v\} \cup (A \setminus O_w)$ or $\{v\} \cup (A \setminus O)$.
As all  these possibilities were tested,
the solution given by the algorithm has the same size as $S$.
As we can test if a graph is split in polynomial time \cite{Golumbic2004}, our algorithm is polynomial.
\end{proof}

Lemma \ref{lemma:solving-special-solutions} helps us solve the case of
special solutions, that is, solutions that intersect exactly one
side of the partition in exactly one vertex.
We now characterize the more general solutions, namely, solutions that
intersect both sides of the partition in more than one vertex.

\begin{lemma}
\label{lemma:bipartite-split-characterization-1}
Let $S$ be a solution with $|S \cap A|>1$ and $|S\cap B|>1$.
Then there exist sets $Q,X\subseteq V(G)$ such that $S = Q \cup X$ and
\begin{itemize}
    \item $Q \in \{\emptyset,K_1,K_2\}$;
    \item $(X,\, V(H)\setminus X)$ is a bipartition of $H = G - Q$;
    \item $N_G(Q)\subseteq X$.
\end{itemize}
\end{lemma}
\begin{proof}
   Write $S_A:=S\cap A$, $S_B:=S\cap B$, $\overline S := V(G)\setminus S$, 
$\overline S_A := A \cap \overline{S}$ and $\overline S_B := B \cap \overline{S}$.
 Let $G' := G\oplus S$, and let $(K,I)$ be a split partition of $G'$.
The following claim is immediate.
\begin{claim}\label{claim:clique-at-most-one-in-I}
If $X$ is a clique in $G'$, then at most one vertex of $X$ lies in $I$.
\end{claim}

Because $G'[S_A]$ and $G'[S_B]$ are cliques, Claim~\ref{claim:clique-at-most-one-in-I} implies that at most one vertex of $S_A$ and at most one vertex of $S_B$ can lie in $I$. Consider two subcases.

\medskip\noindent\textbf{Case 1: } $S\cap I = \varnothing$.\\
Then all vertices of $S$ lie in $K$, so $G'[S]$ is a clique. Thus $S$ is an independent set in $G$. We claim also that $\overline S$ is independent in $G$. Suppose for contradiction there is an edge $uv\in E(G)$ with $u\in \overline{S}_A$ and $v\in \overline{S}_B$. Since vertices of $S_A$ (resp.\ $S_B$) lie in $K$ and $u$ has no neighbors in $S_A$ in $G'$, it follows that $u\in I$; similarly $v\in I$. This contradicts that $I$ is independent. Hence $\overline S$ is independent in $G$, so taking $Q=\varnothing$ and $X=S$ satisfies the lemma  (Figure \ref{fig:bip-split-S-intersects-A-B}a).

\label{fig:bip-split-S-intersects-A-B}
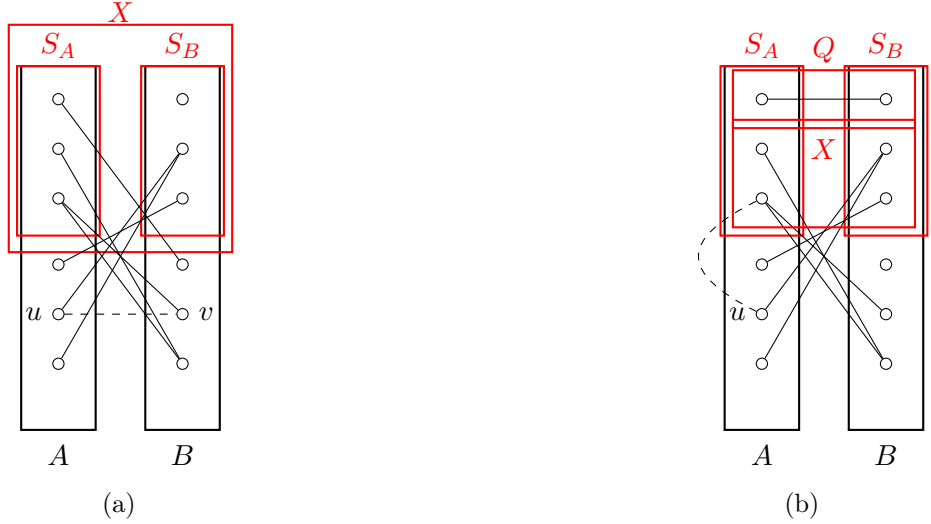
\begin{figure}[h]
\centering

\begin{subfigure}[b]{0.45\textwidth}
\centering
\begin{tikzpicture}[scale=1.1]

\draw[thick] (-1.2,2.2) rectangle (-0.3,-2.2);
\draw[thick] (1.2,2.2) rectangle (0.3,-2.2);

\draw[thick,red] (-1.25,2.2) rectangle (-0.25,0.15);
\node[red] at (-0.75,2.45) {$S_A$};

\draw[thick,red] (0.25,2.2) rectangle (1.25,0.15);
\node[red] at (0.75,2.45) {$S_B$};

\draw[thick,red] (-1.35,2.7) rectangle (1.35,-0.05);
\node[red] at (0,2.85) {$X$};

\node[circle,draw,inner sep=1.5pt] (a1) at (-0.75,1.8) {};
\node[circle,draw,inner sep=1.5pt] (a2) at (-0.75,1.2) {};
\node[circle,draw,inner sep=1.5pt] (a3) at (-0.75,0.6) {};
\node[circle,draw,inner sep=1.5pt] (a4) at (-0.75,-0.2) {};
\node[circle,draw,inner sep=1.5pt,label=left:{$u$}] (a5) at (-0.75,-0.8) {};
\node[circle,draw,inner sep=1.5pt] (a6) at (-0.75,-1.4) {};

\node[circle,draw,inner sep=1.5pt] (b1) at (0.75,1.8) {};
\node[circle,draw,inner sep=1.5pt] (b2) at (0.75,1.2) {};
\node[circle,draw,inner sep=1.5pt] (b3) at (0.75,0.6) {};
\node[circle,draw,inner sep=1.5pt] (b4) at (0.75,-0.2) {};
\node[circle,draw,inner sep=1.5pt,label=right:{$v$}] (b5) at (0.75,-0.8) {};
\node[circle,draw,inner sep=1.5pt] (b6) at (0.75,-1.4) {};

\draw (a1) -- (b4);
\draw (a2) -- (b6);
\draw (a3) -- (b6);
\draw (a3) -- (b5);
\draw (a4) -- (b3);
\draw (a5) -- (b2);
\draw (a6) -- (b2);
\draw[dashed] (a5) -- (b5);

\node at (-0.75,-2.5) {$A$};
\node at (0.75,-2.5) {$B$};

\end{tikzpicture}
\caption{}
\end{subfigure}
\hfill
\begin{subfigure}[b]{0.45\textwidth}
\centering
\begin{tikzpicture}[scale=1.1]

\draw[thick] (-1.2,2.2) rectangle (-0.3,-2.2);
\draw[thick] (1.2,2.2) rectangle (0.3,-2.2);

\draw[thick,red] (-1.25,2.2) rectangle (-0.25,0.15);
\node[red] at (-0.75,2.45) {$S_A$};

\draw[thick,red] (0.25,2.2) rectangle (1.25,0.15);
\node[red] at (0.75,2.45) {$S_B$};

\node[circle,draw,inner sep=1.5pt] (a1) at (-0.75,1.8) {};
\node[circle,draw,inner sep=1.5pt] (a2) at (-0.75,1.2) {};
\node[circle,draw,inner sep=1.5pt] (a3) at (-0.75,0.6) {};
\node[circle,draw,inner sep=1.5pt] (a4) at (-0.75,-0.2) {};
\node[circle,draw,inner sep=1.5pt,label=left:{$u$}] (a5) at (-0.75,-0.8) {};
\node[circle,draw,inner sep=1.5pt] (a6) at (-0.75,-1.4) {};

\node[circle,draw,inner sep=1.5pt] (b1) at (0.75,1.8) {};
\node[circle,draw,inner sep=1.5pt] (b2) at (0.75,1.2) {};
\node[circle,draw,inner sep=1.5pt] (b3) at (0.75,0.6) {};
\node[circle,draw,inner sep=1.5pt] (b4) at (0.75,-0.2) {};
\node[circle,draw,inner sep=1.5pt] (b5) at (0.75,-0.8) {};
\node[circle,draw,inner sep=1.5pt] (b6) at (0.75,-1.4) {};

\draw (a1) -- (b1);
\draw (a2) -- (b6);
\draw (a3) -- (b6);
\draw (a3) -- (b5);
\draw (a4) -- (b3);
\draw (a5) -- (b2);
\draw (a6) -- (b2);
\draw[dashed] 
  (a5) .. controls +(-1,0.5) and +(-1,-0.5) .. (a3);

\draw[thick,red]
    ($(a1)+(-0.35,0.35)$) rectangle ($(b1)+(0.35,-0.35)$);
\node[red] at (0,2.4) {$Q$};

\draw[thick,red]
    ($(a2)+(-0.35,0.35)$) rectangle ($(b3)+(0.35,-0.35)$);
\node[red] at (0,1.2) {$X$};

\node at (-0.75,-2.5) {$A$};
\node at (0.75,-2.5) {$B$};

\end{tikzpicture}
\caption{}
\end{subfigure}

\caption{Situations in the proof of Lemma. \ref{lemma:bipartite-split-characterization-1}.
$(a)$ $S \cap I = \emptyset$. $(b)$ $S \cap I \neq \emptyset$. A dashed line implies the absence of an edge in the original graph.}
\end{figure}

\medskip\noindent\textbf{Case 2: } At least one of $S_A$ or $S_B$ intersects $I$.\\
Let $Q := S\cap I$ and $X := S\cap K$. Recall that $|Q|\le 2$, so $Q$ induces either $K_1$ or $K_2$ ($Q=\varnothing$ is Case 1), proving the first item.
We now show that $(X, V(H)\setminus X)$ is a bipartition of $H:=G-Q$. 
First observe that, as $|S| \geq 4$ and $|Q| \leq 2$, $X \neq \emptyset$. 
Note that $G'[X]$ is a clique (since $X\subseteq K$), so $X$ is independent in $G$. Take any $u\in V(H)\setminus X$. If $u\in A\setminus X$ then $u$ has no neighbors in $S_A\cap X$, hence $u \in I$. Similarly vertices in $B\setminus X$ lie in $I$. Therefore there are no edges with both endpoints in $V(H)\setminus X$, so $(X,V(H)\setminus X)$ is indeed a bipartition of $H$ (Figure \ref{fig:bip-split-S-intersects-A-B}b).
Finally, since $Q\subseteq I$ and vertices of $V(H)\setminus X$ are also in $I$, $Q$ has no neighbors in $V(H)\setminus X$. Thus all neighbors of $Q$ in $G$ (if any) must lie inside $X$, proving $N(Q)\subseteq X$ and completing the proof.
\end{proof}

We now prove a converse statement that will be useful for our algorithm.

\begin{lemma}
\label{lemma:bipartite-split-characterization-2}
Let $Q$ be a clique (of size at most $2$) in $G$.
Let $(X,Y)$ be a bipartition of $H=G-Q$. 
Let $S = X \cup Q$.  
If $N(Q) \subseteq X$, then $S$ is a solution.
\end{lemma}
\begin{proof}
Let $K = X$ and let $I = V(G)\setminus K$.
Since $X$ is an independent set in $G$ and $X \subseteq S$, $K$ is a clique in $G \oplus S$.
Because $Y$ is an independent set in $G$ and $Y \cap S = \emptyset$, $Y$ remains an independent set in $G \oplus S$.
Finally, every vertex of $Q$ has all its neighbors contained in $X$ by assumption.  
Therefore $Q$ has no neighbors in $Y$, and hence $I = Y \cup (Q\setminus X)$ is an independent set in $G \oplus S$.
\end{proof}
We now prove our theorem.

\begin{theorem}
\label{thm:bipartite-split}
Let $\mathcal{C}$ be the class of split graphs.  
Then MSC to $\mathcal{C}$ can be solved in polynomial time if the input is a bipartite graph.
\end{theorem}
\begin{proof}

Let $G$ be a bipartite graph. 
The algorithm begins by testing all sets of size at most $2$ to check whether such a set is a solution, and, if so, returns one of minimum size.
Otherwise, the algorithm has two main phases, according to Lemmas \ref{lemma:solving-special-solutions} and \ref{lemma:bipartite-split-characterization-1}.
In the first phase, the algorithm tests whether there is a special solution that is also optimal, and if that is the case we return such a solution. If the answer is negative in this phase, we proceed to the next phase. By Lemma \ref{lemma:solving-special-solutions}, this can be done in polynomial time.

In the second phase, the algorithm considers all possibilities for $Q\in\{\varnothing,K_1,K_2\}$
 (all cliques of size at most two in $G$).  
Fix such a set $Q$.  
Let $H := G - Q$, and let $C_1,\dots,C_k$ be the components of $H$, which can be computed in polynomial time.  
For each component $C_i$, fix an arbitrary bipartition $(A_i,B_i)$ of $C_i$ such that $|A_i|\le |B_i|$.

If, for some $i$, $N(Q)$ has non-empty intersection with both $A_i$ and $B_i$, then
we discard the choice of $Q$.  
Otherwise we set
\[
W_i :=
\begin{cases}
A_i & \text{if } N(Q)\cap A_i\neq\varnothing,\\[2pt]
B_i & \text{if } N(Q)\cap B_i\neq\varnothing,\\[2pt]
\text{the smaller of }A_i\text{ and }B_i & \text{if }N(Q)\cap C_i=\varnothing.
\end{cases}
\]
We then define $X := \bigcup_{i=1}^k W_i$ and set $S_Q := Q \cup X$.  
The algorithm computes $|S_Q|$ for all possible values of $Q$ and outputs a set $S_{Q^*}$ of minimum cardinality.

We now prove correctness. 
The pair $(X,\, V(H)\setminus X)$ forms a bipartition of $H$, and by construction $N(Q)\subseteq X$.  
Thus Lemma~\ref{lemma:bipartite-split-characterization-2} implies that $S_Q$ is a solution for every choice of $Q$ tested in the second phase of the algorithm.

Let $R$ be an optimal solution. First, suppose that $R$ is a special solution.
By Lemma~\ref{lemma:solving-special-solutions}, as there exists such an $R$, the first phase of the algorithm returns a special solution that is also optimal.
Suppose now that $R$ intersects both $A$ and $B$ at least twice.
By Lemma~\ref{lemma:bipartite-split-characterization-1}, $R$ admits a decomposition  
$R = Q^R \cup X^R$ where $Q^R\in\{\varnothing,K_1,K_2\}$,  
$N(Q^R)\subseteq X^R$, and $(X^R, V(H^R)\setminus X^R)$ is a bipartition of $H^R := G - Q^R$.
Let $C_1,C_2,\ldots ,C_k$ be the components of $H^R$.
As every $C_i$ is connected, for any $i$, $X_i^R := X^R \cap C_i \in \{A_i,B_i\}$.

Consider the iteration of the algorithm corresponding to the choice $Q=Q^R$.  
By (iii) of Lemma~\ref{lemma:bipartite-split-characterization-1}, $N(Q^R) \subseteq X^R$.
Hence $N(Q^R)$ does not intersect both parts of the partition of $C_i$. 
Thus, if $N(Q^R) \cap A_i \neq \emptyset$, then $|W_i|=|A_i|\leq |X_i^R|$, 
if $N(Q^R) \cap B_i \neq \emptyset$, then $|W_i|=|B_i|\leq |X_i^R|$, 
and otherwise $W_i$ is the smaller of $A_i$ and $B_i$, so again $|W_i|\leq |X_i^R|$.
Hence $|X|\leq |X^R|$ and 
\[
|S_{Q^R}| = |Q^R| + |X| \le |Q^R| + |X^R| = |R|.
\]
Since $R$ is optimal, equality must hold, so $S_{Q^R}$ is also optimal.  
The algorithm therefore returns an optimal solution.

Finally, the algorithm tests a polynomial number of possibilities for $Q$, and each iteration runs in polynomial time.  
Therefore, MSC to $\mathcal{C}$ is solvable in polynomial time on bipartite graphs.
\end{proof}

As before, the next corollary is clear.

\begin{corollary}
    \label{cor:co-bipartite-split}
Let $\mathcal{C}$ be the class of split graphs.  
Then MSC to $\mathcal{C}$ can be solved in polynomial time if the input is a co-bipartite graph.
\end{corollary}

\section{Bipartite regular graphs to chordal graphs}
\label{section:bipartite-to-chordals}

A chord in a cycle of a graph is an edge between some vertices of the cycle that is not part of the cycle.
A graph is \emph{chordal} if any cycle of the graph has a chord. Equivalently, a graph is chordal if there exists no induced cycle of length at least four.
In this section, we are interested in solving the MSC problem when the input graph is a bipartite 
regular graph and the target is the class of chordal graphs.
Note that if we complement one of the partitions of a bipartite graph, it results in a chordal graph; hence, bipartite graphs are complementable to chordal graphs.
The next proposition follows directly from the definition of chordal graphs.
\begin{proposition}\label{prop:cycleimpliestriangle}
Let $G$ be a chordal graph with an edge $uv$.
If there exists a cycle $C$ in $G$ that contains $uv$, then there exists a triangle $uvw$ in $G$ with $vw \in E(C)$.
\end{proposition}

We begin this section by characterizing a solution of minimum size for this problem. After that, we apply this characterization to solve the problem in polynomial time for
biregular graphs.

\begin{lemma}\label{lemma:biparte-charact}
Let $G=(A,B,E(G))$ be a 2-connected bipartite graph.
Let $S \subseteq V(G)$ such that $G[S]$ is not a complete bipartite graph.
Then $G' = G \oplus S$ is chordal if and only if the following hold:
\begin{itemize}
    \item[$(i)$] $S$ is a vertex cover in $G$, and
    \item[$(ii)$] $G[S]$ is $2K_2$ induced free.
\end{itemize}
\end{lemma}
\begin{proof}
As $G$ is 2-connected, we have $|V(G)|>2$.
We begin by showing the only if part of the statement.
Thus, let us assume that $G'$ is chordal.
 Suppose by contradiction that $(i)$ is false, that  is, $S$ is not a vertex cover in $G$. 
Then, there exists an edge $uv$ in $G-S$.
As $G$ is 2-connected and $|V(G)|>2$, there exists an induced cycle $C$ in $G$
that contains $uv$.
Suppose for a moment that $C$ is also a cycle in $G'$.
As $G'$ is chordal, either $|C|=3$ or $C$ has a chord, but in the latter case, by Proposition~\ref{prop:cycleimpliestriangle}, there exists a chord of $C$
incident to either $u$ or $v$.
 As $u,v \notin S$, this edge also exists in $G$, a contradiction to the fact that $C$ is an induced cycle in $G$.
Thus, $|C|=3$, a contradiction to the fact that $G$ is bipartite. 

Hence $C$ is not a cycle in $G'$, which implies that $|V(C)  \cap S| > 1$.
Since $G[S]$ is not a complete bipartite graph, it follows that $|S| \neq 2$,
as a graph on two vertices with an edge is a complete bipartite graph.
Consequently, $|S|>2$.

Let $P$ be a maximal subpath of $C$ whose vertices are in $S$, and let $x$ and $y$ be the endpoints of $P$. Let $C_{xy}$ be the subpath of $C$ from $x$ to $y$ avoiding edges of $P$.
As $G[S]$ is not a complete bipartite graph, and $|S|>2$, there exists an edge $ab$ in $G'$ with $a,b \in S$, $a\in A$,
and $b \in B$.
(Note that it can be the case that either $x=a$ or $y=b$).
Now, if $x$ and $y$ are both in the same side of the bipartition, then
$C_{xy} \cdot xy$ is a cycle in $G'$. Otherwise, assuming
$x \in A$ and $y \in B$, we have $C_{xy} \cdot xaby$
is a cycle in $G'$.
In each case we concluded there exists a cycle, say $C'$, in $G'$.
Moreover, let us assume that $|C'|$ is minimum over all such cycles. 
As $G'$ is chordal, this implies that $|C'|=3$.
As $u,v \notin S$, all incident edges to $u$ and $v$
in $G'$ also exist in $G$, and thus $C'$ is also a cycle in $G$, a contradiction to the fact that $G$ is bipartite. This proves $(i)$.
Now, suppose by contradiction that $(ii)$ does not hold, that is, $G[S]$ has an induced $2K_2$, say $ab, a'b'$, with $a,a' \in A$ and $b,b' \in B$.
But then the vertices $a,b,a',b'$ induce a $C_4$ in $G'$, contradicting
that $G'$ is chordal.

We now show the if part of the statement.
Suppose for a contradiction that conditions $(i)$ and $(ii)$ hold, but $G'$ is not chordal.
As $G'$ is not chordal, it has an induced cycle $C$ with $|C| \geq 4$.
We claim that all vertices of $C$ are in $S$.
Indeed, suppose for a moment that there exists a vertex in $C$ not in $S$, say $u$.
Let $u'$ and $u''$ be the neighbors of $u$ in $C$.
As $S$ is a vertex cover, $u',u'' \in S$.
But then $u'u''$ is a chord of $C$, contradicting the assumption that $C$
is an induced cycle.
Now, suppose for a moment that $|C| \geq 5$.
By the pigeonhole principle, at least three vertices of $C$
are in one side of the bipartition, say $A$.
As $|C|>3$, two of these vertices are not adjacent in $C$, so they form a chord in $C$, a contradiction to the fact that $C$ is induced.
Hence, $|C| = 4$. But in that case 
$G'[S]$ contains an induced $C_4$ and,
as $G[S] = \overline{G'[S]}$,
$G[S]$ contains an induced $2K_2$, a contradiction.
\end{proof}

Before focusing on bipartite regular graphs,
we present another lemma that is also valid for any 2-connected bipartite graph.

\begin{lemma}\label{lemma:SsubseteqAthenS=A}
Let $G=(A,B,E(G))$ be a 2-connected bipartite graph. Let $S \subseteq A$.
If $G'=G \oplus S$ is chordal then $S=A$.
\end{lemma}
\begin{proof}
    Suppose by contradiction that 
    $S \neq  A$.
    Let $u \in A \setminus S$.
    As $G$ is 2-connected, there exists an induced cycle in $G$, say $C$, that contains $u$.
    Consider any edge $vw$ in $C$.
    As $G$ is bipartite and $S \subseteq A$, we have that at least one of $v$ or $w$ is not in $S$. Hence, $vw \in E(G')$ and $C$ is a cycle in $G'$.
    Let $ux \in E(C)$.
    By Proposition \ref{prop:cycleimpliestriangle}, as $G'$ is chordal, there exists a triangle in $G'$ that contains $ux$, say $uxy$.
    Recall that $u \notin S$.
    Also, since $G$ is bipartite and $S \subseteq A$, we have $x \notin S$. Hence, as $xy,uy \in E(G')$, we must have $xy,uy \in E(G)$, a contradiction to the fact that $G$ is bipartite.
\end{proof}

$~$
A \textit{biregular graph} is a graph that is both bipartite and regular.
In this section, we present a polynomial-time algorithm for $k$-biregular graphs, that is, biregular graphs in which all vertices have degree~$k$.
We begin by stating a well-known property of biregular graphs, which is a direct corollary of Hall’s Theorem (see, e.g.,~\cite{LovaszPlummer1986}).

\begin{proposition}\label{prop:regularbip-perfectmatch}
Every biregular graph has a perfect matching.
\end{proposition}

From now on, given a graph $G$, we refer to as a subset $S \subseteq V(G)$
as a \emph{solution} if $G \oplus S$ is chordal.
The next proposition is valid for any triangle-free graph.

\begin{proposition}\label{prop:minimality-kreg}
Let $G$ be a triangle-free graph.
Let $S$ be a minimum solution.
If $u \in S$, then $N(u) \cap \overline{S} \neq \emptyset$.
\end{proposition}
\begin{proof}
Suppose by contradiction that all neighbors of $u$ belong to $S$.
Let $G' = G \oplus S$, which is chordal by assumption.
By the minimality of $S$, the graph $G'' = G \oplus (S \setminus \{u\})$ is not chordal.
Hence, $G''$ contains an induced cycle $C$ of length at least four.
Since the only difference between $G'$ and $G''$ concerns adjacencies incident to $u$,
the cycle $C$ must contain $u$.
Let $x$ and $y$ be the neighbors of $u$ in $C$.
Then $x,y \in S \setminus \{u\}$ and $xy \notin E(G'')$.
Therefore, $xy \in E(G)$, which implies that $G$ contains the triangle $uxy$,
a contradiction to the fact that $G$ is triangle-free.
\end{proof}

\begin{lemma}\label{lemma:bipkreg}
Let~$G$ be a~$k$-biregular graph with~$k\geq 2$ and at least $9$ vertices.
Let $(A,B)$ be a bipartition of $G$. Then~$|S|=|A|$, where $S$ is a solution of minimum size.
\end{lemma}
\begin{proof}

 We begin by considering the case when~$G[S]$ is not a complete bipartite graph.
As~$k\geq 2$,~$G$ is~$2$-connected. Hence, by Lemma \ref{lemma:biparte-charact}, as~$G[S]$ is not complete,   ~$S$ is a vertex cover of~$G$. By Proposition \ref{prop:regularbip-perfectmatch},~$G$ has a perfect matching. Note that this matching has size~$|A|$ and, by Kőnig's Theorem \cite{Konig1916}, the size of a minimum vertex cover is also~$|A|$. Hence~$|S| \geq |A|$. Now, note that~$G \oplus A$ is a split graph, and thus also a chordal graph \cite{Golumbic1980}.
    By the minimality of~$S$, we have~$|S| \leq |A|$. This completes the proof.

Now consider the case when~$G[S]$ is a complete bipartite graph.
Let~$a=|A|$,~$r_1 = |A \cap S|$ and~$r_2 = |B \cap S|$,
and suppose without loss of generality that~$r_2 \geq r_1$.
As~$G \oplus A$ is a split graph and thus chordal,
we must have 
\begin{equation}\label{r_1+r_2leqa}
  |S|= r_1+r_2 \leq a.  
\end{equation}

Now, consider the subgraph induced by~$\overline{S}$.
Let us call it~$H$.
Note that~$H$ cannot contain any induced cycle, hence

\begin{equation}\label{EHleq2a-r_1-r_2-1}
|E(H)| \leq |V(H)|-1 = 2a-r_1-r_2-1.
\end{equation}

Also, if~$k\geq r_2+2$, then any vertex in~$H$ has degree at least 2 and contains a cycle.
Thus,~$k \leq r_2+1$, and by Proposition \ref{prop:minimality-kreg},

\begin{equation}\label{k=r_2+1}
k=r_2+1.
\end{equation}

Now, consider the subgraph induced by~$A\setminus S$ and~$S \cap B$.
Let us call it~$Q_1$.
Note that, by \eqref{k=r_2+1},
\begin{equation}\label{EQ1=r_2(r_2-r_1+1)}
|E(Q_1)| = r_2(r_2-r_1+1).
\end{equation}

Analogously, if~$Q_2$ is the subgraph induced by~$S\cap A$ and~$B \setminus S$, we have
\begin{equation}\label{EQ2=r_1(k-r_2)=r_1}
|E(Q_2)|=r_1(k-r_2)=r_1.
\end{equation}

Note that~$|E(H \cup Q_1)|=k(a-r_1) = (r_2+1)(a-r_1)$
and, analogously,~$|E(H \cup Q_2)|=k(a-r_2) = (r_2+1)(a-r_2)$.
Hence, by \eqref{EHleq2a-r_1-r_2-1},
\begin{eqnarray}    
|E(Q_1)|+ |E(Q_2)| &=& |E(H \cup Q_1)|+|E(H \cup Q_2)|-2|E(H)| \nonumber \\
&\geq& (r_2+1)(a-r_1) + (r_2+1)(a-r_2) -2(2a-r_1-r_2-1) \nonumber \\
&=& (r_2+1)(2a-r_1-r_2)-2(2a-r_1-r_2)+2 \nonumber \\
&\geq& (r_2-1)(2a-r_1-r_2)+2  \nonumber \\
&\geq& (r_2-1)(r_1+r_2)+2, \label{EQ1+EQ2geq(r_2-1)(r_1+r_2)+2}
\end{eqnarray}
where the last inequality follows by \eqref{r_1+r_2leqa}.

Now, joining \eqref{EQ1=r_2(r_2-r_1+1)}, \eqref{EQ2=r_1(k-r_2)=r_1} and \eqref{EQ1+EQ2geq(r_2-1)(r_1+r_2)+2}, we obtain that 

\begin{equation}
    (r_1-1)(r_2-1) \leq 0.
\end{equation}

If~$r_1-1 \geq 1$, then~$r_2\leq 1$, a contradiction, because~$r_1 \leq r_2$. Hence,~$r_1\leq 1$.
Thus, by \eqref{EHleq2a-r_1-r_2-1},
$$|E(H)| \leq 2a-r_2-1.$$
Now, as~$r_1 \leq 1$, at most one vertex in~$G[H \cup B]$ has degree~$r_2$.
Thus,~$|E(H)| \geq (r_2+1)(a-r_2)-1$, which implies that

\begin{equation}\label{eq:aleq(r_2+1)+frac1r_2-1}
a \leq (r_2+1) + \frac{1}{r_2-1}.    
\end{equation}

Suppose for a moment that~$r_2 \leq 2$.
In this case, by \eqref{eq:aleq(r_2+1)+frac1r_2-1}
$a\leq 4$ and~$|V(G)| \leq 8$, a contradiction.
Consider now the case when~$r_2 > 2$. 
This implies that~$a\leq r_2+1$ by \eqref{eq:aleq(r_2+1)+frac1r_2-1}.
If~$r_1=1$
then~$|S|=r_1+r_2=1+r_2=a=|A|$, and the proof follows.
And, if~$r_1=0$ then~$r_2 \leq a \leq r_2+1$
If~$a=r_2$, then~$|S|=r_1+r_2=r_2=a=|A|$
and the proof follows.
Hence,~$r_1=0$ and~$a=r_2+1$,
but this is a contradiction to Lemma \ref{lemma:SsubseteqAthenS=A}.
\end{proof}

We are now ready to prove our main theorem.
    
\begin{theorem}
\label{thm:biregular-chordal}
Let $\mathcal{C}$ be the class of chordal graphs.  
Then MSC to $\mathcal{C}$ can be solved in polynomial time if the input is a 2-connected biregular graph.
\end{theorem}
\begin{proof}
    Let~$G$ be a 2-connected~$k$-biregular graph. 
    As $G$ is 2-connected, this implies~$k\geq 2$.
    If~$G$ has fewer than 9 vertices, then we iterate over all subsets of~$V(G)$ using brute force. Otherwise, by Lemma \ref{lemma:bipkreg}, we can find a bipartition for~$G$ in polynomial time and return one of the parts of the bipartition.
\end{proof}

\section{Forests to graphs of degeneracy $k$}
\label{section:forest-deg-k}
Let $\mathcal{F}_k$ denote the class of forests with at least $2k+2$ vertices, for $k \geq 0$.
A graph $G$ is \emph{$k$-degenerate} if every induced subgraph of $G$
has a vertex of degree at most $k$.
The \emph{degeneracy} of $G$ is the minimum integer $k$ for which $G$ is $k$-degenerate.
Let $\mathcal{D}_k$ denote the class of graphs of degeneracy exactly $k$, for $k \geq 0$.
We show that every graph in $\mathcal{F}_k$ can be complemented into $\mathcal{D}_k$,
and moreover, that MSC to $\mathcal{D}_k$ can be solved in polynomial time
when the input graph belongs to $\mathcal{F}_k$.
The following characterization is immediate from the definition of degeneracy.

\begin{proposition} \label{prop:GinmathcalD_kthenaandb}
A graph $G$ belongs to $\mathcal{D}_k$ if and only if
\begin{itemize}
    \item[(a)] every induced subgraph of $G$ has a vertex of degree at most $k$, and
    \item[(b)] there exists an induced subgraph of $G$ in which every vertex has degree at least $k$.
\end{itemize}
\end{proposition}

From now on, given a graph $G$, we refer to a subset $S \subseteq V(G)$
as a \emph{solution} if $G \oplus S \in \mathcal{D}_k$.

\begin{proposition} \label{prop:Sgeqk}
Let $k \geq 2$.
Let $G$ be a forest. If $S$ is a solution, then $|S| \geq k$.
\end{proposition}

\begin{proof}
As $G \oplus S \in \mathcal{D}_k$, by Proposition~\ref{prop:GinmathcalD_kthenaandb}$(b)$
there exists an induced subgraph $H'$ of $G \oplus S$ such that every vertex of $H'$
has degree at least $k$.
Let $H = G[V(H')]$. Clearly, $H$ is a forest, and hence $H$ has a vertex,
say $v$, of degree at most one in $H$.
Since the degree of $v$ in $H'$ is at least $k$, the vertex $v$ is adjacent
to at least $k$ vertices of $H'$ in $G \oplus S$.
At most one of these adjacencies can come from $H$, and therefore
$v$ must gain at least $k-1$ neighbors via the complementation on $S$.
This is only possible if $v \in S$ and at least $k-1$ other vertices of $H'$
also belong to $S$.
Hence, $|S| \geq k$.
\end{proof}

A set of independent vertices in a graph is called a \emph{sibling set} if they have a common neighbor. 
The vertices in a sibling set are called \emph{siblings}.

\begin{lemma}
\label{lemma:forest-to-degeneracy}
Every graph in $\mathcal{F}_k$ is complementable to $\mathcal{D}_k$ for any fixed $k\geq 2$.
 Let $G \in \mathcal{F}_k$
and $S$ be a minimum solution.
    If $G$ has a sibling set of size $k$,
    then $|S|=k$. Otherwise, $|S|=k+1$.
\end{lemma}
\begin{proof}
Let $G \in \mathcal{F}_k$.
Suppose for a moment that $G$ has a sibling set of size $k$, say $S$.
We will show that $G'=G \oplus S$ has degeneracy $k$. As every vertex in $G'[S \cup \{u\}]$, where $u$
 is the common neighbor of the sibling set,
has degree $k$, part $(b)$ of Proposition \ref{prop:GinmathcalD_kthenaandb} follows.
Let $X \subseteq V(G')$. We need to show that there exists a vertex in $G'[X]$ with degree at most $k$.
If $X \subseteq S$, then any vertex in $G'[X]$ has degree at most $k$, so part $(a)$ of Proposition \ref{prop:GinmathcalD_kthenaandb} holds.
Otherwise, $X$ has a vertex in $\overline{S}$.

Note that $G'[X \cap \overline{S}]$ is a forest.
If $|X \cap \overline{S}|=1$, then the only vertex in
$X \cap \overline{S}$ has at most $k$ neighbors in $G[X \cap S]$, and therefore also in $G'[X]$ and part (a) of Proposition \ref{prop:GinmathcalD_kthenaandb} follows.  
Hence, we may assume that $|X \cap \overline{S}|>1$, which implies
that $G'[X \cap \overline{S}]$ has two vertices, say $v$ and $w$, of degree at most one.
Suppose by contradiction that none of them has degree at most $k$ in $G'[X]$. In that case, both vertices have degree at least $k+1$. As both vertices have degree at most one in $G'[X \cap \overline{S}]$, it implies that both $v$ and $w$ are adjacent, in $G'$, to all vertices in $S$.
Note that this implies that, in $G$, they are also adjacent to all vertices in $S$, a contradiction to the fact that $G$ is a forest since $k\geq 2$.
Now, by Proposition \ref{prop:Sgeqk}, $S$ is a minimum solution and the proof follows.

Finally, suppose that $G$ does not have a sibling set of size $k$. As $G$ has at least $2k+2$ vertices, it has an independent set, say $S$ of size $k+1$.
We will show that $G'=G \oplus S$ is a graph with degeneracy $k$. As every vertex in $G'[S]$ has degree $k$, then part $(b)$ of Proposition \ref{prop:GinmathcalD_kthenaandb} follows.
Let $X \subseteq V(G')$. We need to show that there exists a vertex in $G'[X]$ with degree at most $k$.
If $X \subseteq S$, then any vertex in $G'[X]$ has degree at most $k$, so part $(a)$ of Proposition \ref{prop:GinmathcalD_kthenaandb} holds.
Otherwise, $X$ has a vertex in $\overline{S}$.
Let $v$ be a vertex of minimum degree in $G'[X \cap \overline{S}]$. As $G'[X \cap \overline{S}]$ is a forest, such a vertex has degree at most one in
$G'[X \cap \overline{S}]$.
As there is no subset of $k$ siblings in $G$,
$v$ is adjacent to at most $k-1$ vertices in $S$.
Hence, $v$ has at most $k$ neighbors in $G'[X]$ and part $(a)$ of Proposition \ref{prop:GinmathcalD_kthenaandb} follows.
Next, we will show that $S$ is a minimum solution.

Suppose first that $k=2$. Then $G$ has at least 6 vertices.
 As $G$ does not have a sibling set of size 2, and $G$ has no cycles, 
each component of $G$ is either a $K_2$ or a $K_1$.
Then for any subset $S$ of size 2, in $G'=G\oplus S$,
Proposition \ref{prop:GinmathcalD_kthenaandb}(b) is not true.
Therefore, the size of the minimum solution is 3, i.e., $k+1$.

Consider now the case $k\geq 3$.
Suppose by contradiction that there exists $R \subseteq V(G)$ such that $G'=G \oplus R \in \mathcal{D}_k$ and $|R| < |S| = k+1$.
By Proposition \ref{prop:Sgeqk}, $|R|=k$.
By Proposition \ref{prop:GinmathcalD_kthenaandb}(b), there exists
    $X \subseteq V(G')$  such that every vertex in $G'[X]$ has degree at least $k$, which implies 
    that $X \setminus R \neq \emptyset$.
    We will show that $G' [X \cap \overline{R}]$ has no isolated vertices.
For a contradiction, suppose $G' [X \cap \overline{R}]$ has an isolated vertex $v$.
Since $|R|=k$ and any vertex in $G'[X]$ has degree at least $k$, $v$ is 
adjacent to all the vertices of $R$ in $G'[X]$, which implies that $R$ is a sibling set in $G$, a contradiction.
Note that 
$G' [X \cap \overline{R}]$ is a forest.
Let $u$ and $v$ be two vertices of degree 1 in the same component of $G' [X \cap \overline{R}]$.
Then both $u$ and $v$ are adjacent to at least $k-1$ vertices of
$R$ in $G'[X]$. Since $k\geq 3$, $u$ and $v$ have a common neighbor in $X\cap R$, 
which implies there exists a cycle in $G$, a contradiction.
\end{proof}

We prove Theorem \ref{thm:forest-degenerate} with the help of Lemma \ref{lemma:forest-to-degeneracy}.

\begin{theorem}
\label{thm:forest-degenerate}
Let $\mathcal{D}_k$ be the class of graphs with degeneracy $k$ ($k \geq 0$).  
Then MSC to $\mathcal{D}_k$ can be solved in polynomial time if the input is a forest.
\end{theorem}
\begin{proof}  
We distinguish cases according to the value of $k$.
Suppose first that $k=0$. In this case, the class $\mathcal{D}_0$ consists precisely of edgeless graphs. If the input forest $G$ is a set of isolated vertices, then the empty set is a minimum solution.
Otherwise, if $G$ contains an edge, then $G \oplus S$ contains an edge for every non-empty $S$,
except when $G$ is isomorphic to $K_2$ or $K_2 \cup nK_1$ ($n\geq 1$),
which can be checked directly.
Hence, the MSC problem can be decided in polynomial time in this case.

Now suppose that $k=1$.
In this case, the class $\mathcal{D}_1$ consists exactly of forests that contain at least one edge.
If the input forest $G$ contains an edge, then $G \in \mathcal{D}_1$ and the empty set is a solution.
If $G$ consists of a single isolated vertex, then no solution exists.
Otherwise, $G$ has at least two isolated vertices, and complementing any two of them creates an edge,
yielding a graph of degeneracy~$1$.
Thus, the MSC problem can be solved in polynomial time.

Finally, suppose that $k \geq 2$.
By Proposition~\ref{prop:Sgeqk}, if $|V(G)| < k$, then $G$ is not complementable to $\mathcal{D}_k$,
and we answer No.
If $k \leq |V(G)| \leq 2k+1$, then $G$ has constant size (for fixed $k$),
and we can test all subsets of vertices to find a minimum solution in polynomial time.
Now assume that $|V(G)| \geq 2k+2$, that is, $G \in \mathcal{F}_k$.
We can verify in polynomial time whether $G$ has a sibling set of size $k$
by iterating over all vertices and checking whether some neighborhood has size at least $k$.
If such a set exists, we return it.
Otherwise, since $G$ is bipartite, we compute a maximum independent set in polynomial time
and return any subset of size $k+1$ of this set.
Correctness follows from Lemma~\ref{lemma:forest-to-degeneracy}.
This completes the proof.
\end{proof}

\section{Increasing the connectivity of a graph}
\label{section:increasing-connectivity}
A graph \( G \) is said to be \emph{\(k\)-connected} if \( |V(G)| > k \) and the removal of any set of fewer than \( k \) vertices leaves the graph connected.  
The \emph{connectivity} of \( G \), denoted by \( \kappa(G) \), is the largest integer \( k \) for which \( G \) is \(k\)-connected. Equivalently, \( \kappa(G) \) is the minimum number of vertices whose removal disconnects \( G \) or reduces it to a trivial graph. If \( G \) is disconnected, we define \( \kappa(G) = 0 \).

Let \( G \) be a connected graph.

A vertex \( v \in V(G) \) is called a \emph{cut vertex} (or \emph{articulation point}) if the removal of \( v \) increases the number of connected components of \( G \).
A
\emph{block} is a maximal connected subgraph without a cutvertex.
The \emph{block--cut tree} (or \emph{BC-tree}) of \( G \) is the bipartite graph \( T_G \) defined as follows.  
The vertex set of \( T_G \) consists of one node for each block of \( G \) and one node for each cut vertex of \( G \).  
A cut vertex \( v \) is adjacent in \( T_G \) to every block node \( B \) such that \( v \in V(B) \).

The main target of this section is to complement a graph into a 2-connected graph. From now on, given a graph $G$, a \emph{solution} is a subset $S \subseteq V(G)$ such that $G \oplus S$ is 2-connected.

\begin{lemma} \label{lemma:kappa1-leaves-TG}
Let~$G$ be a graph with $|V(G)| \geq 3$ and $\kappa(G)=1$.
The size of a minimum solution equals the number of leaves of $T_G$.
\end{lemma}
\begin{proof}
    We begin by showing that $|S|$ is at least the number of leaves of $T_G$. For this, it suffices to show that $S$ contains at least one vertex of every block of $G$ whose corresponding node is a leaf in $T_G$.
    Suppose by contradiction that this is not the case. Then there exists a block $B$ in $G$, whose corresponding node is a leaf in $T_G$ and $S \cap V(B) = \emptyset$.
    Let $v$ be the node adjacent to $B$ in $T_G$, which exists because $|V(G)| \geq 3$ and $\kappa(G)=1$. By the definition of $T_G$, $v$ is a cut vertex of $G$ and also $v\in B$. Since $B$ is a block, it contains at least two vertices, so there exists
$u \in V(B)\setminus\{v\}$. As $v$ is a cut vertex, there exists a vertex $w$ in $G-v$ with no path to $u$.
    As $S \cap V(B)= \emptyset$, there is also no path from $u$ to $w$ in $(G \oplus S)-v$, so $v$ is a cut vertex in $G \oplus S$, a contradiction.

We now show that $|S|$ is at most the number of leaves of $T_G$.
For this, let $S$ be the set formed by choosing an arbitrary non-cut vertex $v_B$ in each block $B$
whose corresponding node is a leaf of $T_G$.
We claim that $G' = G \oplus S$ is $2$-connected. 
As $|V(G')|\geq 3$,
it suffices to show that $G' - x$ is connected for every $x \in V(G')$.
Let $u,w \in V(G-x)$, and let $B_1$ and $B_2$ be the blocks of $G$ such that
$u \in B_1$ and $w \in B_2$.
If $B_1 = B_2$, then $u$ and $w$ are connected by a path inside $B_1$,
and all edges of $B_1$ are preserved in $G'$, so they are connected in $G'-x$.

Otherwise, consider two leaf blocks $Y$ and $Z$ of $T_G$ such that
there exists a path in $T_G$ from $B_1$ to $Y$ and a path from $B_2$ to $Z$
that avoids the node corresponding to $x$ (if $x$ is a cut vertex).
These paths correspond to paths in $G-x$ from $u$ to $v_Y$ and from $w$ to $v_Z$.
The vertices $v_Y$ and $v_Z$
are adjacent in $G'$. Hence, $u$ and $w$ are connected in $G'-x$.
This proves that $G'$ is $2$-connected.
\end{proof}

The block–cut tree can be obtained in linear time using Tarjan’s depth-first search algorithm \cite{tarjan1972dfs}.
From Lemma \ref{lemma:kappa1-leaves-TG}, we have the next result.

\begin{theorem}
\label{thm:to-2conn}
Let $\mathcal{C}$ be the class of $2$-connected graphs.  
Then MSC to $\mathcal{C}$ can be solved in polynomial time.
\end{theorem}

\begin{proof}
If $|V(G)| \leq 3$ then we solve the problem by inspection, so let us assume that $V(G) \geq 4$.
Suppose first that the input graph $G$ is connected.
We compute the block--cut tree $T_G$ of $G$.
If $T_G$ consists of a single node, then $G$ is already $2$-connected and we return the empty set.
Otherwise, we return
$S=\{v_B : B \text{ is a leaf of } T_G\}$,
which is optimal by Lemma~\ref{lemma:kappa1-leaves-TG}.

Now suppose that $G$ is not connected.
We compute the block--cut tree of each connected component of $G$.
We iterate over all such components and ask if such 
block--cut tree has at least two leaves.
If it is the case, then we select a vertex $v_B$ for each leaf $B$ of the corresponding block-cut tree and add it to $S$; otherwise,
we select $\min \{2,|C|\}$ arbitrary vertices for each component $C$ and add it to set $S$.
Hence, at least two vertices are in $S$ for each component of size more than one.
Also, when $|C| \geq 3$, by repeating the same argument given in Lemma~\ref{lemma:kappa1-leaves-TG}, the subgraph of $G \oplus S$ induced by $C$, which we can call $C \oplus S$, is $2$-connected 

Let $u$ be a vertex in a component $C$ of $G$.
Suppose that $|C|\geq 2$. We select two arbitrary vertices $u_1$ and $u_2$ in $C \cap S$.
As $C \oplus S$ is 2--connected when $|C|\geq 3$, by Menger's Theorem \cite{BondyM08}, there exists two internally disjoint paths from $u$ to $u_1$ and $u_2$, we call these paths $P_u$ and $Q_u$, respectively.
Note that when $u \in S$, one of these paths consists on a single vertex.
When $|C|=1$ we abuse notation and set $u_1=u_2=u$.
In this case, we also set $P_u=Q_u=u$.

We first prove that $G \oplus S$ is $2$-connected.
For this, it suffices to use Menger's Theorem \cite{BondyM08} and prove that for every pair of vertices $u,v$, there exists two internally disjoint paths from $u$ to $v$ in $G \oplus S$.
Suppose first that $u$ and $v$ belongs to the same component $C$ of $G$.
If $|C| \geq 3$ then, as already mentioned $C \oplus S$
is $2$-connected and the statement follows. If $|V(C| = 2$, then both $u$ and $v$ are in $S$.
As $|V(G)| \geq 4$, we must have at least $2$ vertices in $S\setminus \{u,v\}$, say $x$ and $y$. Thus $uxv$ are $uyv$ are 2 internally disjoint paths from $u$ to $v$ in $G \oplus S$.

Now suppose $u$ and $v$ are in different components of $G$, say $C$ and $D$.
Note that $P_uP_v$ and $Q_uQ_v$ are paths in $G \oplus S$, which are distinct if at least one of $C$ and $D$ have size more than $1$. Now, if $|C|=|D|=1$, as $V(G) \geq 4$, there are at least $2$ vertices in $S \setminus \{u,v\}$, say $x$ and $y$, and $uxv$ and $uyv$ are 2 internally disjoint paths
 
 We now prove that $S$ is a minimum solution.
Suppose, for a contradiction, that there exists a set $R$ such that
$G \oplus R$ is $2$-connected and $|R| < |S|$.
Then there exists a component $H$ of $G$ such that
$|R \cap V(H)| < |S \cap V(H)|$.
Note that $|R\cap V(H)|>0$, otherwise, as $|V(G)|>1$, $G \oplus R$
is not 2-connected, a contradiction. 
This implies that $|H| \geq 2$ and so $|S_H|=2$ and $|R_H|=1$. But then $R \cap V(H)$
is an articulation in $G \oplus R$, again a contradiction.

\end{proof}

\section{Disconnecting a graph}

In this section we study the problem of disconnecting a graph.
From now on, given a graph $G$, a \emph{solution} is a subset
$S \subseteq V(G)$ such that $G \oplus S$ is disconnected.
We begin by proving two basic properties.

\begin{proposition}\label{prop:monotone-in-induced}
Let $G$ be a connected graph and 
let $S$ be a solution for $G$.
If $S \subseteq A$ then
 $S$ is also a solution for $G[A]$.
\end{proposition}
\begin{proof}
Let $G' = G \oplus S$ and $G'' = G[A] \oplus S$.
Since $G$ is connected but $G'$ is disconnected, every connected component of $G'$ contains a vertex of $S$ and thus also of $A$. Let $C$ and $D$ be two distinct connected components of $G'$, and let $u \in V(C)\cap A$ and $v \in V(D)\cap A$.
Suppose, for a contradiction, that $G''$ is connected. Then there exists a path $P$ from $u$ to $v$ in $G''$.
Since $S\subseteq A$, the operation $\oplus S$ affects only adjacencies between vertices of $A$. Hence,
$G''=(G\oplus S)[A]=G'[A].
$
Therefore, every edge of $P$ also belongs to $G'$, implying that $P$ is a path from $u$ to $v$ in $G'$. This contradicts the fact that $u$ and $v$ belong to different connected components of $G'$.
\end{proof}

\begin{proposition}\label{prop:S-intersects-every-component}
Let $G$ be a connected graph and let $S$ be a solution.
Then $S$ intersects every component of $G \oplus S$.
\end{proposition}

\begin{proof}
Suppose by contradiction that there exists a component $C$ of $G \oplus S$
such that $V(C) \cap S = \emptyset$.
Since $G$ is connected, there exists an edge $uv \in E(G)$ with
$u \in V(C)$ and $v \notin V(C)$.
Let $D$ be the component of $G \oplus S$ containing $v$; clearly $D \neq C$.
As $u \notin S$, the operation $\oplus S$ does not affect the edge $uv$,
and hence $uv \in E(G \oplus S)$.
Therefore, $u$ and $v$ belong to the same component of $G \oplus S$,
contradicting the assumption that $C$ and $D$ are distinct.
\end{proof}

For developing our algorithm, we rely on the following definition.
A \emph{split} of a graph $G$ (not to be confused with a split graph)
is a partition of $V(G)$ into four sets $(A_2, A_1, B_1, B_2)$ such that
every vertex of $A_1$ is adjacent to every vertex of $B_1$, and
there are no edges between $A_2$ and $B_2$.
From now on, given such a split, we let $A := A_1 \cup A_2$
and $B := B_1 \cup B_2$. (Figure \ref{fig:split}a).
From now on, we fix a graph $G$ and a split $(A_2,A_1,B_1,B_2)$ of $G$.

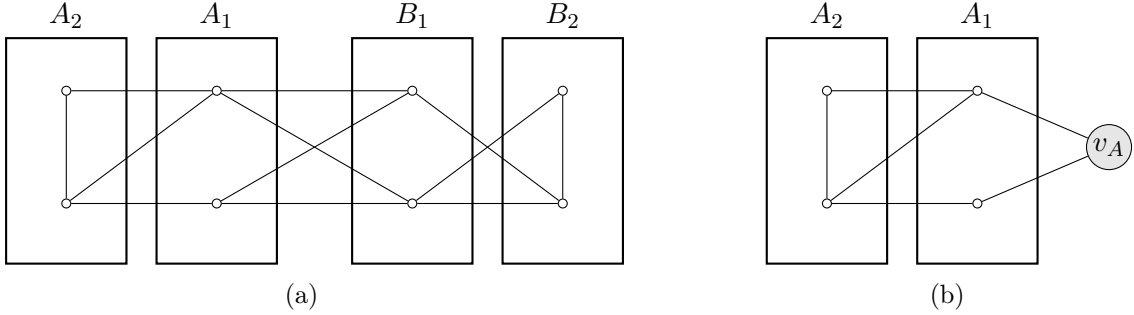
\begin{figure}[h]
\centering

\begin{subfigure}{0.48\textwidth}
\centering
\begin{tikzpicture}[
    scale=1,
    every node/.style={circle, draw, inner sep=1.2pt},
    box/.style={draw, thick, rectangle}
]

\draw[box] (-0.8,-0.8) rectangle (0.8,2.2);
\draw[box] (1.2,-0.8) rectangle (2.8,2.2);
\draw[box] (3.8,-0.8) rectangle (5.4,2.2);
\draw[box] (5.8,-0.8) rectangle (7.4,2.2);

\node (a21) at (0,1.5) {};
\node (a22) at (0,0) {};
\node (a11) at (2,1.5) {};
\node (a12) at (2,0) {};
\node (b11) at (4.6,1.5) {};
\node (b12) at (4.6,0) {};
\node (b21) at (6.6,1.5) {};
\node (b22) at (6.6,0) {};

\draw (a21)--(a22);
\draw (b21)--(b22);
\draw (a21)--(a11);
\draw (a22)--(a12);
\draw (a22)--(a11);
\draw (b11)--(b22);
\draw (b12)--(b21);
\draw (b12)--(b22);

\draw (a11)--(b11);
\draw (a11)--(b12);
\draw (a12)--(b11);
\draw (a12)--(b12);

\node[draw=none] at (0,2.5) {$A_2$};
\node[draw=none] at (2,2.5) {$A_1$};
\node[draw=none] at (4.6,2.5) {$B_1$};
\node[draw=none] at (6.6,2.5) {$B_2$};

\end{tikzpicture}
\caption{ }
\end{subfigure}
\hfill
\begin{subfigure}{0.48\textwidth}
\centering
\begin{tikzpicture}[
    scale=1,
    every node/.style={circle, draw, inner sep=1.2pt},
    box/.style={draw, thick, rectangle}
]

\draw[box] (-0.8,-0.8) rectangle (0.8,2.2);   
\draw[box] (1.2,-0.8) rectangle (2.8,2.2);    

\node (a21) at (0,1.5) {};
\node (a22) at (0,0) {};

\node (a11) at (2,1.5) {};
\node (a12) at (2,0) {};

\node[fill=gray!20] (vA) at (3.75,0.75) {$v_A$};

\draw (a21)--(a22);
\draw (a21)--(a11);
\draw (a22)--(a12);
\draw (a22)--(a11);

\draw (vA)--(a11);
\draw (vA)--(a12);

\node[draw=none] at (0,2.5) {$A_2$};
\node[draw=none] at (2,2.5) {$A_1$};

\end{tikzpicture}
\caption{ }
\end{subfigure}

\caption{$(a)$ A split $(A_2,A_1,B_1,B_2)$ of a graph $G$. $(b)$ The graph $G_A$.}
\label{fig:split}
\end{figure}

The next lemma shows important properties for a valid solution.

\begin{lemma}\label{lemma:Sdisconnected-characterization}
Let $S$ be a solution for $G$.
Then
\begin{itemize}
    \item $S \cap A_2 = \emptyset$ or $S \cap B_2 = \emptyset$
    \item If $S \cap B_2 = \emptyset$ then $S$ is also a solution for $G[A\cup B_1]$. If $S \cap A_2  = \emptyset$ then $S$ is also a solution for $G[B\cup A_1]$.
  \item One of the next possibilities is true: $S \subseteq A$, $S \subseteq B$,
$A_1 \subseteq S $ or $B_1 \subseteq S$.
\end{itemize}
\end{lemma}
\begin{proof}
We first show the first statement.
 Suppose, for contradiction, that there exist $u,v\in S$ with 
$u\in A_2$ and $v\in B_2$.
Since $uv\notin E(G)$, we have $uv\in E(G \oplus S)$.
Hence $u$ and $v$ lie in the same component of $G\oplus S$, say $C$.
    As $G \oplus S$ is not connected, there exists a component in $G \oplus S$, say $D$, distinct from $C$.
    Let $w$ be a vertex in $D$.
    As, in $G$, $v$ has no neighbors in $A$, we have $w \notin A$. Indeed, otherwise $wv \notin E(G)$ and $wv  \in E(G \oplus S)$, a contradiction.
    Hence $w \in B$. But $u\in A_2$ has no neighbors in $B$ in $G$, and the same reasoning 
shows that $uw\in E(G\oplus S)$, again contradicting $w\in D$.

For the second statement, assume that $S \cap B_2 = \emptyset$.
Since $S \subseteq A \cup B_1$, 
the claim follows from Proposition~\ref{prop:monotone-in-induced}.
For the third statement, suppose that $S \not\subseteq A$ and $S \not\subseteq B$.
By the first statement, $S \cap A_2 = \emptyset$ or
$S \cap B_2 = \emptyset$.
Assume for the moment that $S \cap B_2 = \emptyset$.
We claim that $A_1 \subseteq S$ or $B_1 \subseteq S$.
Suppose by contradiction that there exist
$u \in A_1 \setminus S$ and $v \in B_1 \setminus S$.
Let $C$ be the component of $G \oplus S$ containing $v$.
Since $v \notin S$, for every $a \in A_1$ we have $av \in E(G)$ and hence
$av \in E(G \oplus S)$; thus $A_1 \subseteq V(C)$.
Likewise, since $u \notin S$, every vertex of $B_1$ is adjacent to $u$ in
$G\oplus S$, and therefore $B_1 \subseteq V(C)$.

Since $G \oplus S$ is disconnected, there exists another component $D \neq C$.
By Proposition~\ref{prop:S-intersects-every-component}, $D$ contains a vertex
$x \in S$.
Because $A_1 \cup B_1 \subseteq V(C)$, this vertex $x$ cannot lie in $A_1$ or
$B_1$, so $x \in A_2 \cup B_2$.
But our assumption $S \cap B_2 = \emptyset$ forces $x \in A_2$.
Since $S \not\subseteq A$, there exists a vertex $y \in S \cap B_1$.
As $x \in A_2$ and $y \in B_1$, we have $xy \notin E(G)$, and therefore
$xy \in E(G \oplus S)$.
This contradicts the fact that $x \in V(D)$ and $y \in V(C)$.
Thus $A_1 \subseteq S$ or $B_1 \subseteq S$, completing the proof (Figure).
\end{proof}

The main purpose of introducing the notion of a split is to solve the problem
recursively by reducing it to smaller subgraphs of the original input \cite{Rao2008,Cicerone1999,Cunningham1982}.
For this purpose, we define $G_A$ as the graph whose vertex set consists of $A$
together with a special vertex, denoted by $v_A$, which represents all vertices of
$B_1$. The adjacencies in $G_A$ are defined as follows: the subgraph induced by
$A$ is exactly $G[A]$, and the special vertex $v_A$ is adjacent to every vertex of
$A_1$ (Figure \ref{fig:split}). Symmetrically, we define $G_B$ by exchanging the roles of $A$ and $B$.
From now on, we use $G_A,G_B, v_A$ and $v_B$ as defined before.

\begin{lemma} \label{lemma:SinAB1-SinA1}
Let $S$ be a solution for $G[A \cup B_1]$.
\begin{itemize}
 \item If $B_1 \cap S = \emptyset$, then $S$ is also a solution for $G_A$. 
 \item If $B_1 \subseteq S$, then $(S \setminus B_1) \cup \{v_A\}$ is a solution for $G_A$. 
\end{itemize}
\end{lemma}
\begin{proof}
Identify one vertex of $B_1$ with $v_A$
For the first item, we have that
$A \cup \{v_A\}$ is a subset of vertices of $G[A \cup B_1]$ and the proof follows by Proposition~\ref{prop:monotone-in-induced}.
Now, for the second item, let $S_A= (S \setminus B_1) \cup \{v_A\}$ and suppose by contradiction
that $G''= G_A \oplus S_A$ is not connected.
Let $G'=G[A \cup B_1] \oplus S$ and note that  $G'' \subseteq G'$. Hence all vertices of $V(G'')$ lie in a connected component of $G'$, say $C$. So, as $G'$ is not connected, there exists a vertex $u \in V(G') \setminus V(G'')$ that has no neighbors in $C$. This implies, as $(A_2,A_1,B_1,B_2)$ is a split, that $A_1 \subseteq S$. But then $G''$ is not connected, a contradiction.
\end{proof}

\begin{lemma}\label{lemma:Rdisconnected-thenSdisconnected}
Let $R$ be a solution for $G_A$.
If $v_A \notin R$ then $R$ is a solution for $G$.
If $v_A \in R$ then $S$ is a solution for $G$, where $S = (R \setminus \{v_A\}) \cup B_1$.
\end{lemma}

\begin{proof}

Since $G_A \oplus R$ is disconnected, 
there exist two components $C$ and $D$ in $G_A \oplus R$.
By Proposition~\ref{prop:S-intersects-every-component}, 
there exist $x,y \in R$ such that $x \in V(C)$ and $y \in V(D)$.

\textbf{Case 1: $v_A \notin R$.}

Suppose, towards a contradiction, that $G \oplus R$ is connected.
As $v_A \notin R$, both $x$ and $y$ lie in $A$.
Because $G \oplus R$ is connected, there exists a path $P$ from $x$ to $y$ in $G \oplus R$.
If all internal vertices of $P$ lie in $A$, then $P$ is also a path in $G_A \oplus R$, contradicting that $x$ and $y$ lie in different components.  
Hence at least one internal vertex of $P$ lies outside $A$, and thus in $B_1$.
Let $P_x$ be the maximal subpath of $P$ starting at $x$ whose vertices lie in $A$,  
and let $x'$ be the end of $P_x$ distinct from $x$.
Similarly, we define $P_y$ and $y'$.  
Since $x',y'\in A_1$, and $v_A$ is adjacent to all vertices of $A_1$ in $G_A$, 
\[
    P_x \cdot x' v_A y' \cdot P_y
\]
is a path from $x$ to $y$ in $G_A \oplus R$, contradicting again that $x$ and $y$ lie in different components.
This proves the first statement.

\medskip
\textbf{Case 2: $v_A \in R$.}

Assume towards a contradiction that $G \oplus S$ is connected.
As $v_A \in R$, we may assume without loss of generality that $x = v_A$ (otherwise, we proceed as in Case 1).
Because $G \oplus S$ is connected, there exists a path $P$ in $G \oplus S$ whose ends are $y$ and some vertex $z \in B_1$.
Choose $P$ among such paths so that it has minimum length.
By minimality of $P$, all internal vertices of $P$ lie in $A$.
Since $v_A$ is adjacent to every vertex of $A_1$ in $G_A$, 
\[
    P - z + v_A
\]
is a path in $G_A \oplus R$ joining $y$ to $x=v_A$.
But $x$ and $y$ lie in different components of $G_A \oplus R$, a contradiction. 
\end{proof}

A split is \emph{trivial} if $|A|=1$ or $|B|=1$.
We say that a graph is \emph{prime} if it admits only trivial splits.
For a vertex $u$ of a graph, we let $N[u] = \{u\} \cup N(u)$.

\begin{lemma}\label{lemma:split-prime}
Let $G$ be a connected prime graph and let $S$ be a solution for $G$.
Then $S = N[u]$ for some $u \in V(G)$.
\end{lemma}

\begin{proof}
Since $G \oplus S$ is disconnected, let $C$ and $D$ be two distinct connected
components of $G \oplus S$.
By Proposition~\ref{prop:S-intersects-every-component}, both
$C$ and $D$ contain vertices of $S$.
Set
\[
A_1 = V(C)\cap S,\quad A_2 = V(C)\setminus S,\qquad
B_1 = V(D)\cap S,\quad B_2 = V(D)\setminus S .
\]

Note that $ab \notin E(G \oplus S)$ for any $a \in V(C)$ and $b \in V(D)$.
Since $A_1,B_1 \subseteq S$ and $A_2 \cap S = B_2 \cap S = \emptyset$,
it follows that $(A_2,A_1,B_1,B_2)$ is a split of $G$.
As $G$ is prime, this split must be trivial.
Without loss of generality, we may assume that $A_2 = \emptyset$ and $|A_1| = 1$.
Thus $S = N[a]$, where $a$ is the unique vertex in $A_1$.
\end{proof}

Finally, we can prove the main result of this section.

\begin{theorem}\label{thm:to-0conn}
Let $\mathcal{C}$ be the class of disconnected graphs.
Then MSC to $\mathcal{C}$ can be solved in polynomial time.
\end{theorem}

\begin{proof}
We solve a weighted version of the problem: given a positive weight function
$w$ on $V(G)$, we seek a minimum--weight solution $S$.
If $G$ is already disconnected, we return the empty set.

If $G$ is prime, then by Lemma~\ref{lemma:split-prime}, every feasible solution
has the form $N[u]$ for some vertex $u$.
We simply test all vertices and return the set of minimum weight.

Assume now that $G$ admits a nontrivial split $(A_2,A_1,B_1,B_2)$.
We recurse on the two graphs obtained from the split.
Let
\[
G_A = G[A\cup\{v_A\}] \quad\text{with } w(v_A)=w(B_1),
\qquad
G_B = G[B\cup\{v_B\}] \quad\text{with } w(v_B)=w(A_1).
\]
Let $R_1$ and $R_2$ be the solutions returned for $G_A$ and $G_B$, respectively.

From $R_1$ we construct a feasible solution for $G$ as follows:
if $v_A\notin R_1$, set $S_1=R_1$;
if $v_A\in R_1$, set $S_1=(R_1\setminus\{v_A\})\cup B_1$.
In both cases, $w(S_1)=w(R_1)$.
Analogously, we construct $S_2$ from $R_2$.
By Lemma~\ref{lemma:Rdisconnected-thenSdisconnected},
both $G\oplus S_1$ and $G\oplus S_2$ are disconnected.
We return the solution of minimum weight.

We now prove correctness.
Let $S$ be a minimum--weight set such that $G\oplus S$ is disconnected.
By Lemma~\ref{lemma:Sdisconnected-characterization},
$S\cap A_2=\emptyset$ or $S\cap B_2=\emptyset$.
Assume without loss of generality that $S\cap B_2=\emptyset$.
Then, again by Lemma~\ref{lemma:Sdisconnected-characterization},
$G[A\cup B_1]\oplus S$ is disconnected, and one of the following holds:
\[
\text{(i) } S\subseteq A,\qquad
\text{(ii) } A_1\subseteq S,\qquad
\text{(iii) } B_1\subseteq S.
\]

If $S\subseteq A$, then by Lemma~\ref{lemma:SinAB1-SinA1},
$G_A\oplus S$ is disconnected, and thus the recursive solution satisfies
$w(R_1)\le w(S)$.
Since $G\oplus S_1$ is disconnected, minimality of $S$ implies
$w(S)\le w(R_1)$ and $w(S)\le w(R_2)$.
Hence $w(S)=w(R_1)=\min\{w(R_1),w(R_2)\}$.

If $A_1\subseteq S$, then Lemma~\ref{lemma:SinAB1-SinA1} implies that
\[
G_B\oplus\bigl((S\setminus A_1)\cup\{v_B\}\bigr)
\]
is disconnected, and therefore $w(R_2)\le w(S)$.
Again, by minimality, $w(S)\le w(R_1)$ and $w(S)\le w(R_2)$, which yields
$w(S)=w(R_2)=\min\{w(R_1),w(R_2)\}$.
The case $B_1\subseteq S$ is symmetric.

Therefore, the algorithm always returns a minimum--weight solution.
Each recursive call reduces the number of original vertices by at least one.
Moreover, detecting whether $G$ is prime and, if not, producing a nontrivial
split can be done in polynomial time \cite[Problem~1]{Cunningham1982}.
Thus, the overall running time is polynomial.
\end{proof}

\section*{Ethics declaration}
Ethics declaration: not applicable.

\section*{Funding}
Juan Gutiérrez was supported by Fondo Semilla UTEC 2025.

\bibliographystyle{plain}
\bibliography{main}

\end{document}